\documentclass[letterpaper, 10 pt, conference]{ieeeconf}
\IEEEoverridecommandlockouts
\usepackage{cite}
\usepackage{amsmath,amssymb,amsfonts}
\usepackage{algorithmic}
\usepackage{graphicx}
\usepackage{textcomp}
\usepackage{xcolor}
\usepackage[draft,inline,nomargin,index]{fixme}
\newcommand{\R}{\mathbb{R}}

\newcommand{\np}{\vspace{7mm}}

\newcommand\figref{Figure~\ref}
\newcommand\secref{Section~\ref}

\newcommand\lref{Lemma~\ref}
\newcommand\assumref{Assumption~\ref}
\newcommand\eqqref{Equation~\eqref}

\newtheorem{theorem}{Theorem}

\newtheorem{assumption}{Assumption}
\newtheorem{lemma}{Lemma}

\newcommand{\RR}{\mathbb{R}} 
\newcommand{\pmu}{\mu}           
\newcommand{\prior}{m}           
\newcommand{\pvar}{\sigma^2}     
\newcommand{\nn}{\sigma_n^2}     
\newcommand{\kf}{\kappa}
\newcommand{\km}{\boldsymbol{\mathcal{K}}}
\newcommand{\kv}{\boldsymbol{\kappa}}
\newcommand{\DS}{\mathcal{DS}}
\newcommand{\HH}{\mathcal{H}}

\newcommand{\ys}{\boldsymbol{y}}
\newcommand{\xs}{\boldsymbol{x}}

\newcommand{\nx}{{n_x}}

\newcommand{\lp}{\left(}
\newcommand{\rp}{\right)}
\newcommand{\T}{^{\top}}

\newcommand{\hext}{\bar{h}}

\newtheorem{remark}{Remark}

\usepackage{todonotes}

\usepackage{tikz}
\usetikzlibrary{positioning}

\definecolor{forest}{HTML}{009933}
\definecolor{bluesea}{HTML}{6666ff}
\definecolor{orange}{HTML}{ff8c1a}
\definecolor{rederror}{HTML}{ff3333}
\definecolor{brown}{HTML}{A52A2A}
\fxsetup{theme=color,mode=multiuser}
\FXRegisterAuthor{lg}{lgenv}{\color{forest}LG} 
\FXRegisterAuthor{bt}{btenv}{\color{bluesea}BT} 
\FXRegisterAuthor{fs}{fsenv}{\color{orange}FS} 
\FXRegisterAuthor{lm}{lmenv}{\color{rederror}LM} 

\title{\LARGE \bf
Data-Driven Analytic Differentiation via High Gain Observers and Gaussian Process Priors
}

\author{Biagio Trimarchi$^{1}$, Lorenzo Gentilini$^{1}$, Fabrizio Schiano$^{2}$, and Lorenzo Marconi$^{1}$%
\thanks{$^{1}$ Biagio Trimarchi, Lorenzo Gentilini, and Lorenzo Marconi are with the Center for Research on Complex Automated Systems (CASY), Department of Electrical, Electronic and Information Engineering (DEI), University of Bologna, Bologna, Italy
	{\tt\small e-mails: \{biagio.trimarchi2, lorenzo.gentilini6, lorenzo.marconi\}@unibo.it}}%
\thanks{$^{2}$Fabrizio Schiano is with Leonardo S.p.a., Leonardo Labs, Rome, Italy
	{\tt\small e-mail: fabrizio.schiano.ext@leonardo.com}}%
}

\begin{document}
\maketitle
\thispagestyle{empty}
\pagestyle{empty}

\begin{abstract}
The presented paper tackles the problem of modeling an unknown function, and its first $r-1$ derivatives, out of scattered and poor-quality data. The considered setting embraces a large number of use cases addressed in the literature and fits especially well in the context of control barrier functions, where high-order derivatives of the safe set are required to preserve the safety of the controlled system.
The approach builds on a cascade of high-gain observers and a set of Gaussian process regressors trained on the observers' data. The proposed structure allows for high robustness against measurement noise and flexibility with respect to the employed sampling law. 
Unlike previous approaches in the field, where a large number of samples are required to fit correctly the unknown function derivatives, here we suppose to have access only to a small window of samples, sliding in time.
The paper presents performance bounds on the attained regression error and numerical simulations showing how the proposed method outperforms previous approaches.
\\
	
\end{abstract}

\section{Introduction}

Autonomous systems have gained a lot of interest in the last decades, and we witness each year a big effort to increase their autonomy and capabilities. This effort was motivated by both an increase in computational resources and a reduction in Size, Weight, Power, and Cost (SWaPC) of such systems.
The consequent advancements in autonomous systems technologies unlocked the use of data-driven techniques on real systems. In the field of data-driven techniques, Gaussian Process (GP) regression \cite{rasmussen_gaussian_2006} is gaining popularity thanks to its non-parametric nature, the analytical tractability, and the existence of analytical bounds on the estimate error \cite{lederer2019uniform}. This property makes them particularly suitable for safety-critical applications, where uncertainty and noisy information could lead to a critical failure.
In the same context, another impactful advancement of the last years is the so-called \emph{control barrier functions}~\cite{ames2019control}, which are able to act as a filter for the control input of an autonomous system to guarantee that the safety requirements are always satisfied.~\cite{ames2019control}. Applications of control barrier functions can be seen in various contexts such as: quad-copters teleoperation \cite{xu2018safe}, multi-robot systems~\cite{wang2017safety}, and adaptive cruise control \cite{taylor2020adaptive}.
Recently, researchers tried to merge together GPs and CBFs giving birth to a new learning control paradigm \cite{khan_gaussian_2022}. Such a solution succeeds in all those cases when an analytic formulation of the safe-set is a priori not known, e.g. in the case of exploration of an unknown environment.
The main drawback of this last approach is that an overestimate of the barrier function could compromise the safety of the system. Moreover, the barrier condition relies on the knowledge of the barrier functions derivatives, that, in this learning scenario, is difficult to retrieve. As a matter of fact, GPs regressors suffer from loss of information during differentiation ~\cite{holsclaw2013gaussian} which makes them not suitable for such an application. 
On the other hand, an accurate estimate of time derivative can be generated using High Gain Observers (HGOs), which provide practical convergence even in the case of model uncertainty for high enough gains~\cite{tornambe1992high}.

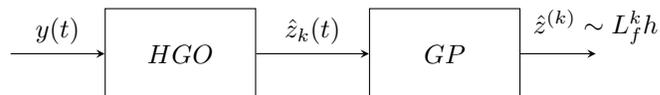
\begin{figure}
    \centering
\begin{tikzpicture}
\node [draw,
    fill=white,
    minimum width=2cm,
    minimum height=1.2cm,
    right=1cm 
]  (controller) {$HGO$};
 
\node [draw,
    fill=white, 
    minimum width=2cm, 
    minimum height=1.2cm,
    right=1.5cm of controller
] (system) {$GP$};
 

\draw[-stealth] (-0.25,0) -- (controller.west) 
    node[midway,above]{$y(t)$};
 
\draw[-stealth] (controller.east) -- (system.west)
    node[midway,above]{$\hat{z}_k(t)$};

\draw[-stealth] (system.east) -- ++ (1.0,0) 
    node[midway](output){}node[right,above]{$\hat{z}^{(k)} \sim L_f^k h$};
    
\end{tikzpicture}
\caption{Structure of the proposed approach: the high gain observer generates the data needed to fit the Gaussian process.}
\label{fig:approach}
\end{figure}

Motivated by these works, in this paper we take a step back from control barrier functions and focus on proposing a novel approach to estimate the derivative of an unknown function of which we have scarce measurements. We estimate the derivative of this function by combining Gaussian processes and high gain observers  \cite{tornambe1992high} as depicted in \figref{fig:approach}.
The idea of combining of HGOs and GPs is not novel since it was already proposed in \cite{buisson2021joint}, however, our overarching goal is different. In \cite{buisson2021joint} the system dynamics are predicted out of collected data. Instead, in this paper, we focus mainly on reproducing a state-dependent unknown function, and its derivatives, regressing only on scattered and noisy measurements.
In summary, our contribution is a technique to obtain an analytic estimate of the directional derivative of an unknown function out of very scattered and noisy data.
We show that the proposed approach has provable convergence guarantees and we compare our solution, through numerical simulations, to the naive approach of deriving a regressor fitted directly to the scarcely available measurements.  
We chose the context of autonomous systems and CBFs to offer the reader an example of a situation in which  our approach could be adopted. However, we highlight that our approach is general and could be applied to any context in which one wants to estimate the derivative of an unknown function from scarce measurements.
The paper unfolds as follows. Section \ref{SEC:PRELIMINARIES} reviews the basics of Gaussian processes inference and the theory of high gain state observation. 
Section \ref{SEC:PROBLEM} describes the general problem along with our assumptions and our proposed approach. \ref{SEC:SIMULATIONS} describes numerical simulations to corroborate our approach and \ref{SEC:CONCLUSIONS} concludes the paper and describes some future work.


\section{Preliminaries}%
\label{SEC:PRELIMINARIES}
\subsection*{Notation}
Consider a nonlinear system of the form $\dot{x} = f(x)$, with state $x \in \R^n$. Moreover consider a function $h(x) : \R^n \to \R$, we denote with $\mathcal{L}_f h(x)$ the \emph{Lie derivative} $\mathcal{L}_f h(x) = \frac{\partial h}{\partial x}f(x)$. 
\par The operator $\lVert \cdot \rVert : \R^n \to \R$ denotes the standard Euclidean norm. The set denoted as $\mathcal{B} (x) = \left \{ \bar{x} \in \R^n : \lVert \bar{x} - x \rVert \leq 1 \right \}$ is the unit ball centered in $x \in \R^n$. Moreover, if $w : \R^+ \to \R$, we set 
    ${\lVert w \rVert}_\infty = 
    \max_{t \geq 0} \lVert w(t) \rVert$.

\subsection*{Gaussian Process Regression}
Let $x \in \mathcal{X} \subset \R^{n_x}$. A GP is a stochastic process such that any finite number of outputs is assigned a joint Gaussian distribution with a prior mean function
$\prior: \RR^\nx \mapsto \RR$ and covariance defined through the kernel $\kf: \RR^\nx \times \RR^\nx \mapsto \RR$ \cite{rasmussen_gaussian_2006}.
While there are many possible choices of mean and covariance functions, in this work we keep the formulation of $\kf$ general, with
the only constraint expressed by~\assumref{assm:K-CONTINUOUS-BOUNDED} below. 
Thus, when we assume that a function $f : \mathcal{X} \subset \R^{n_x} \to \R$ is described by a Gaussian process with mean $\prior$ and covariance $\kf$ we write
\begin{equation*}
   f \sim \mathcal{GP} \lp \prior ( \cdot ), \kf \lp \cdot, \cdot \rp \rp.
\end{equation*}
In the following we force, without loss of generality, $\prior \lp x \rp = 0_\nx$ for any $x \in \mathcal{X}$.
\par Let us denote a time window of $N \in \mathbb{N}$ time instants $t_k \in \RR_{>0}$ with $\mathcal{S} = \{t_1, t_2, \ldots, t_N \}$. Supposing to have access to a data-set of samples $\DS = \{\lp x(t_k), y(t_k) \rp \in \mathcal{X} \times \RR, t_k \in \mathcal{S} \}$,
with each pair $\lp x, y \rp \in \DS$ obtained as $y(t_k) = f(x(t_k)) + \varepsilon(t_k)$ with
$\varepsilon(t_k) \sim \mathcal{N}(0, \nn I_{1 \times r})$ white Gaussian noise with known variance $\nn$, the regression is performed by conditioning
the prior GP distribution on the training data $\DS$ and a test point $x$.
Denoting $\xs = (x(t_1), \dots, x(t_N))\T$ and $\ys = (y(t_1), \dots, y(t_N))\T$, the conditional posterior distribution of $f$, given the data-set, is still a
Gaussian process with mean $\pmu$ and variance $\pvar$ given by \cite{rasmussen_gaussian_2006}
\begin{equation}%
   \label{eq:GP-POSTERIOR}
   \begin{split}
      \pmu \lp x \rp & = \kv \lp x \rp\T \lp \km + \nn I_{N} \rp^{-1} \ys, \\
      \pvar \lp x \rp & = \kf(x, x) - \kv(x)\T\left(\km + \nn I_{N}\right)^{-1}\kv(x),
   \end{split}
\end{equation}
where $\km \in \RR^{N \times N}$ is the \textit{Gram matrix} whose $(k,h)$-th entry is $\km_{k,h} = \kf(\xs_k, \xs_h)$,
with $\xs_k$ the $k$-th entry of $\xs$, and $\kv(x) \in \RR^N$ is the kernel vector whose $k$-th component is $\kv_k(x) = \kf(x, \xs_k)$.
\begin{remark}
   The assumption of measurements perturbed by Gaussian noise is commonly used in learning-based control since it is caused, for example,
   by numerical differentiation (see~\cite{umlauft2017feedback}) 
\end{remark}
From now on we suppose that the following standing assumptions hold (see\cite{buisson2021joint} ,~\cite{lederer2021uniform})
\begin{assumption}%
	\label{assm:MU-CONTINUOUS-BOUNDED}
	The funciton $\pmu_i$ is Lipschitz continuous with Lipschitz constant $L_{\pmu}$, and its norm is bounded by $\pmu_{\text{max}}$.
 \end{assumption}
\begin{assumption}%
   \label{assm:K-CONTINUOUS-BOUNDED}
   The kernel function $\kf$ is Lipschitz continuous with constant $L_{\kf}$, with a locally Lipschitz derivative of constant $L_{d\kf}$, and
   its norm is bounded by $\kf_{\text{max}}$.
\end{assumption}
Although any kernel fulfilling~\assumref{assm:K-CONTINUOUS-BOUNDED} can be a valid candidate, in the following,
we exploit the commonly adopted \textit{squared exponential kernel} as prior covariance function, which can be expressed as
\begin{equation}%
   \label{eq:EXPONENTIAL-KERNEL}
   \kf(x, x') = \np \exp \lp -\lp x - x' \rp \T \Lambda^{-1}  \lp x - x'\rp \rp
\end{equation}
for all $x, x' \in \RR^\nx$, where $\Lambda = \text{diag}(2\lambda_{x_1}^2, \dots, 2\lambda_{x_\nx}^2)$, $\lambda_{x_{i}} \in \RR_{>0}$ is
known as \emph{characteristic length scale} relative to the $i$-th signal, and $\np$ is usually called \emph{amplitude} \cite{rasmussen_gaussian_2006}.
\begin{assumption}%
	\label{assm:HH-CONTINUOUS-BOUNDED}
	Each component of the unknown map $f$ has a bounded norm in the RKHS\footnote{Reproducing Kernel Hilbert Space}  $\HH$ generated to the kernel $\kf$, in~\eqqref{eq:EXPONENTIAL-KERNEL}.
\end{assumption}
\begin{remark}
   \assumref{assm:HH-CONTINUOUS-BOUNDED} is asking some Lipschitz continuity property of the unknown function that makes it well-representable by means of a Gaussian process prior.
   Nevertheless, it represents a very strong assumption, difficult to be checked even if the unknown function is known.
   \assumref{assm:HH-CONTINUOUS-BOUNDED} can be relaxed to the condition that each component $\hext_i$ is a sample from the Gaussian process $\mathcal{GP} \lp 0, \kf \lp \cdot, \cdot \rp \rp$,
   which, in turn, leads to a larger pool of possible unknown functions and it is easier to be checked. As an example, the pool generated by
   the squared exponential kernel~\eqqref{eq:EXPONENTIAL-KERNEL} is equal to the space of continuous functions.
\end{remark}

\newcommand{\kfs}{\mathcal{X}}
\newcommand{\NN}{\mathbb{N}}
\newcommand{\Bb}{\mathcal{B}}
\newcommand{\by}{\boldsymbol{y}}
{We recall a result based on~\cite{lederer2021uniform}.}
\begin{lemma}%
   \label{lem:MEAN-BOUND}
   {
   Consider a zero-mean Gaussian process defined through a kernel  $\kf: \kfs \times \kfs \mapsto \RR$, satisfying~\assumref{assm:K-CONTINUOUS-BOUNDED}
   on the compact set $\kfs$. Furthermore, consider a continuous unknown function $f: \kfs \mapsto \RR$ with Lipschitz constant $L_f$, and $N \in \NN$ observations
   $y^i = f \lp x^i \rp + \varepsilon^i$, with $\varepsilon^i \sim \mathcal{N}(0, \nn I_{n_{y}})$.
   Then the posterior mean $\pmu$ and posterior variance $\pvar$ conditioned on the training data
   $\DS = \left\{ \lp x^1, y^1 \rp, \dots, \lp x^N, y^N \rp \right\}$
   are continuous with Lipschitz constants $L_{\pmu}$ and $L_{\pvar}$ on $\kfs$, respectively, satisfying}
   \begin{equation*}
      \begin{split}
         & L_{\pmu} \le L_{\kf} \sqrt{N} \left\| \lp \km + \nn I_N \rp^{-1} \by \right\|, \\
         & L_{\pvar} \le 2 \rho L_{\kf} \lp 1 + N \left\| \lp \km + \nn I_N \rp^{-1} \right\| \max_{x, x' \in \kfs} \kf(x,x') \rp,
      \end{split}
   \end{equation*}
   {
   with $x = (x^1, \dots, x^N)\T$ and $\by = (y^1, \dots, y^N)\T$.
   Moreover, pick $\delta \in \lp 0,1 \rp$, $\rho > 0$ and set}
   \begin{equation*}
      \begin{split}
         & \beta \lp \rho \rp = 2 \log \lp \frac{M \lp \rho, \kfs \rp}{\delta} \rp, \\
         & \alpha \lp \rho \rp = \lp L_{f} + L_{\pmu} \rp \rho + \sqrt{\beta \lp \rho \rp L_{\pvar} \rho},
      \end{split}
   \end{equation*}
  {
   with $M \lp \rho, \kfs \rp$ the $\rho$-covering number}~\footnote{{The minimum number such that there exists a set $\mathcal{X}_\rho$ so that its cardinality is equal to $M \lp \rho, \kfs \rp$ and
   $\max_{x \in \kfs} \min_{x' \in \mathcal{X}_\rho} \left\| x - x' \right\| \le \rho$.}} {
   related to the set $\kfs$.
   Then, the bound}
   \begin{equation*}
      \left| f(x) - \pmu(x) \right| \le \sqrt{\beta \lp \rho \rp} \pvar \lp x \rp + \alpha \lp \rho \rp \hspace*{0.5cm} \forall x \in \kfs
   \end{equation*}
   {
   holds with probability at least $1 - \delta$.
   }
\end{lemma}

\par Before concluding this section, let us denote with $\mathcal{E}_{f, \mathcal{S}} (x) : \mathcal{X} \to \R$ the regressor fitted with the Gaussian process to the data set $\mathcal{DS}$. In this work, we set $\mathcal{E}_{f, \mathcal{S}} (x) = \mu(x)$. 

\subsection*{High Gain Observers}
Let $z \in Z \subset \R^n$ be the state of an autonomous linear system written in canonical observability form
\begin{align}
    \begin{cases}
	\dot{z} = A z + d(t) \\
	y = C z + v(t),
    \end{cases}
    \label{z_system}
\end{align}

where $y \in \R$ is the measured output of the system, $d \in \R^n$ is a bounded disturbance, $v \in \R$ is the measurement noise and $F, G$ and $H$ have the form
\begin{align*}
    A &= \left[ 
	\begin{matrix}
	    0_{(n-1) \times 1} & I_{n-1} \\ 
	    0 & 0_{1 \times (n-1)}
	\end{matrix}
	\right]
    \\
    C &= \left[ 
	\begin{matrix}
	    1 & 0_{1 \times (n-1)} 
	\end{matrix}
	\right].
	\quad \quad \quad \quad 
\end{align*}

The problem of state observation for the system (\ref{z_system}) can be solved by the high-gain observer
\begin{align}
    \dot{\hat{z}} = A \hat{z} + 
		    D_l K (C \hat{z} - y),
    \label{z_observer}
\end{align}

where $\hat{z} \in \R^n$ is the state of the observer,
$K = {\left[ 
	\begin{matrix}
	    k_1 & k_2 & \cdots & k_n
	\end{matrix}
    \right]}^T$
is a vector of positive parameters ($k_i>0$) chosen so that the matrix $A + KC$ is Hurwitz and 
$D_l = \text{diag}(l, l^2, \ldots, l^n)$ is a diagonal matrix parameterized by $l > 0$.

\par In this framework we recall a result from \cite{tornambe1992high}.

\begin{lemma}%
    \label{lem:observer}
    Given the system (\ref{z_system}) with observer (\ref{z_observer}), if $d$ and $v$ are bounded, then there exists $l^* \in \R$ such that for every $l > l^*$ there exists $c_1, c_2, c_3, c_4 > 0$ such that for every $t>0$ we have
    \begin{align*}
	\lvert {\hat{z}}_i(t) - z_i(t) \rvert &\leq 
	c_1 l^{i-1} e^{-c_2 l t }\lvert 
	{\hat{z}}_i(0) - z_i(0) \rvert \\
	&+ \frac{c_3}{l^{n+1-i}} {\lVert d \rVert}_\infty +
	c_4 l^{i-1} {\lVert v \rVert}_\infty
    \end{align*}
\end{lemma}


\section{Problem Set-up and Main Results} \label{SEC:PROBLEM}
Consider a nonlinear system of the form
\begin{equation}
    \begin{matrix}
	\dot{x} = f(x) + g(x)u,
	&
	y = h(x) + \varepsilon(t), 
    \end{matrix}
    \label{x_dynamics}
\end{equation}
where $x \in \mathcal{X} \subset \R^{n_x}$, $u \in \mathcal{U} \subset \R^{n_u}$, $f : \mathcal{X} \to \mathcal{X}$ and $g : \mathcal{X} \to \R^{n_x \times n_u}$, $h : \mathcal{X} \to \R$, and $\varepsilon(t) : \R^+ \to \R$ is a measurement noise. Let us assume that $f, g, h$ are smooth functions of the state. Let the initial state $x(0) = x_0$ and the control input $u(t)$ be fixed and let $\Phi : \R^+ \to \mathcal{X}$ be the resulting solution of (\ref{x_dynamics}). We introduce the set 
\begin{align}
\mathcal{T}_{\delta, T_1, T_2} = \left \{ x \in \R^n_x\; : \;   x \in \bigcup_{t \in [T_1, T_2]} \delta\, \mathcal{B}(\Phi(t)) \right \}
\end{align}

where $T_2 > T_1 > 0$ and $\delta > 0$.
We suppose that both $x(t)$ and $y(t)$ are measurable for each $t>0$. In this framework, we are interested in obtaining estimates of $h$ and its first $r-1$ functional derivatives (i.e. $h$, $L_f^1 h$, \ldots, $L_f^{r-1} h$),  along the trajectory of $\eqref{x_dynamics}$, where $r \in \mathbb{N}$, $r <n$, is given. The estimates are functions $\hat{h}^{(k)}: {\cal X} \to \mathbb{R}$, $k=0,\ldots, r-1$, to be computed so that $\lvert \hat{h}^{(k)}(x) - L_f^k h(x) \rvert$ is small in some sense. We suppose the first $r-1$ time derivatives are not affected by output, namely, we assume the following.
\begin{assumption}
    $L_g L_f^{k} h(x) = 0$ for each $k < r$.
\end{assumption}
\begin{assumption}%
    \label{assm:gaussian_distributed}
    $h$ is a realization of a Gaussian process
    \begin{align*}
    h \sim \mathcal{GP}( 0, \kf_0 \lp \cdot, \cdot \rp )\,.
    \end{align*}
\end{assumption}
    Following \cite{rasmussen_gaussian_2006}, the previous assumption implies that also the higher derivatives of $h$ are realizations of Gaussian processes with \ certain covariance $\kf_k$, namely
\begin{align*}
    L_f^k h \sim \mathcal{GP}( 0, \kf_k \lp \cdot, \cdot \rp ), \ \forall k = 0, 1, \ldots r-1.
\end{align*}
Let $\mathcal{S}^k_j = \left \{ t^k_{j-(N-1)}, t^k_{j-(N-2)}, \ldots t^k_j \right \}$ be a sliding time window of $N \in \mathbb{N}$ time instants $t^k_i \in \R^+$, where $t_j > t_{j-1}$. The strategy presented later tunes the Gaussian process linked to $\hat h^{(k)}$  with a data set obtained by sampling the available measures $x(t)$, $y(t)$ and the state of the dirty derivative observer introduced later, at the time instances in $\mathcal{S}^k_j$, $j >0$.  The window is updated when the value of the state $x(t)$ fulfills $\lVert x(t) - x(t^k_j) \rVert > \tau > 0$, with  $\tau>0$, taking $t^k_{j+1}=t$. 
\par Ideally, setting $\hat{h}^{(k)} = \mathcal{E}_{L_f^k, \mathcal{S}^k_j}$ would guarantee a probabilistic bound on the estimation error\cite{lederer2019uniform}\cite{lederer2021uniform}. However, the values of $L_f^k$ are not measurable and thus we cannot construct the needed dataset for training. A first option could be to set $\hat{h} = \mathcal{E}_{h, \mathcal{S}^0_j}$ and to take $\hat{h}^{(k)} = L_f^k \hat{h}$ as estimates for $L_f^k h$. This approach, however, has no theoretical guarantees, and it also leads to an accumulation of errors in the process \cite{solak_derivative_nodate}. Moreover, any uncertainty of the system dynamics $f$ would also compromise the quality of the estimate. In this work, we rather propose a technique to model $\hat{h}^{(k)}$ fulfilling
\begin{align*}
     \lvert \hat{h}^{(k)}(x) - \mathcal{E}_{L_f^k, \mathcal{S}^k_j}(x) \rvert < \epsilon
\; \;\; \forall x \in \mathcal{X}
\end{align*}
where $\epsilon > 0$ is a bound which depends on the noise on the available data, and not relying on the knowledge of the vector field $f$. Then, we use this property to compute a probabilistic bound of convergence of this estimate to $L_f^k h$.
\par The core of the proposed approach is to use an high gain observer to  generate an approximation of $L_f^k h$. Following the structure of (\ref{z_observer}) 
\begin{equation}
\begin{aligned}
    \dot{\hat{z}}_1 &= \hat{z}_2 + l k_1 (\hat{z}_1 - y(t)) \\
    \dot{\hat{z}}_2 &= \hat{z}_3 + l^2 k_2 (\hat{z}_1 - y(t)) \\
    &\vdots \\ 
    \dot{\hat{z}}_r &=  l^r k_r (\hat{z}_1 - y(t)),
\end{aligned}
\label{EQ:hat}
\end{equation}
where $k_i$, $i=1,\ldots, r$, are chosen as in \secref{SEC:PRELIMINARIES}. 
\par By using the property that the state $\hat z_{k+1}$ of (\ref{EQ:hat}) practically converges to $L_f^k h$, we compute the estimate $\hat{h}^{(k)}$ as $ \mathcal{E}_{\hat z_{k+1}, \mathcal{S}^k_j}$. The following theorem can be then proved.
\setcounter{theorem}{0}
\begin{theorem} \label{THEO:estimate}
Let  \assumref{assm:K-CONTINUOUS-BOUNDED} holds. Then there exist $\bar{t} > 0$ and $l^*>0$ such that for each $t > \bar{t}$ and $l > l^*$, there exist constants $c_1, c_2, c_3 > 0$ such that 
\begin{align*}
        \lvert \hat{h}^{k}(x) - \mathcal{E}_{L_f^k h, \mathcal{S}^k_j}(x) \rvert \leq c_1 N \max 
        \{
            c_2 l^{k} \lVert \varepsilon(t) \rVert_{\infty}, 
            c_3 l^{k-r}
        \} 
    \end{align*}
for all $x \in \mathcal{X}$.
\end{theorem}
\begin{proof} 
Let $Y = {\left [ \begin{matrix}
    z_{t_{j-(N-1)}}, \cdots, z_{t_{j}}    
        \end{matrix} 
\right ]}^T$ and $\hat{Y} = {\left [ \begin{matrix}
    \hat{z}_{t_{j-(N-1)}}, \cdots, \hat{z}_{t_{j}}    
        \end{matrix}
\right ]}^T$
then 
\begin{align}
    \lvert \mathcal{E}_{\hat{z}_{k+1}, \mathcal{S}^k_j}(x) - \mathcal{E}_{L_f^k h, \mathcal{S}^k_j}(x) \rvert = \\
    = \lvert \kv(x)^T (\km + \sigma_n^2 I_N)^{-1} (\hat{Y} - Y) \rvert \leq \\
    \leq \lVert \kv(x) \rVert \lVert (\km + \sigma_n^2 I_N)^{-1} \rVert \lVert \hat{Y} - Y \rVert   \,.
\end{align}
Because of  \assumref{assm:K-CONTINUOUS-BOUNDED}, we have $\lVert \kv(x) \rVert \leq \kv_{max}$. It follows that there exists a $ K > 0 $ so that $\lVert (\km + \sigma_n^2 I_N)^{-1} \rVert \leq K$. Let us set $c_1 = \kv_{max} K$, then
\begin{align*}
    \lvert \mathcal{E}_{\hat{z}_{k+1}, \mathcal{S}^k_j}(x) - \mathcal{E}_{L_f^k h, \mathcal{S}^k_j}(x) \rvert \leq c_1 \lVert \hat{Y} - Y \rVert \leq \\
    \leq c_1 \sum_{i = j-{N+1}}^{j} \lVert \hat{z}_{k+1}(t_i) - z_{k+1}(t_i) \rVert
\end{align*}
Be means of \lref{lem:observer}, there exist a time instant $\bar{t} > 0$ so that, for all $t_j > \bar{t}$ the following holds
\begin{align*}
\lvert \mathcal{E}_{\hat{z}_{k+1}, \mathcal{S}^k_j}(x) - \mathcal{E}_{L_f^k h, \mathcal{S}^k_j}(x) \rvert \leq \\ \leq c_1 N \max 
        \{
            c_2 l^{k} \lVert \varepsilon(t) \rVert_{\infty}, 
            c_3 l^{k-r}
        \}, 
        \forall x \in \mathcal{X}
\end{align*}
from which the result follows.
\end{proof}

{Theorem \ref{THEO:estimate} yields a bound on the difference between the ideal and the proposed estimate. \lref{lem:MEAN-BOUND}, then, can be used to establish a probabilistic guarantee of convergence of $\hat{h}^{(k)}$ to $L_f^k h$ on the set $\mathcal{T}_{\delta, t_\epsilon, t}$, where $t_\epsilon$ is an arbitrarily small time and $t$ is the current time, as formalized in the following theorem.}

\begin{theorem}\label{THEO:bound}
Pick $\eta \in (0, 1)$ and $\rho > 0$. Let $W_k^j$, $W_{\mu_k^j}$ and $W_{\sigma^{2j}_k}$ be   the  Lipschitz constants of, respectively, $L_f^kh$,  of the mean $\mu_k^j$ and variance $\sigma_k^{2,j}$ of the Gaussian process linked to {$L_f^kh$} on the set $\mathcal{T}_{\delta,t_{j-(N+1)}, t_j}$. Furthermore, let 

\begin{align*}
    \beta(\rho) = 2 \log{(\frac{M(\rho, \mathcal{T}_{0,t_{j-(N+1)}, t_j})}{\eta})} \\
    \alpha(\rho) = (W_k^j + W_{\mu_k^j}) \rho + \sqrt{\beta(\rho) W_{\sigma^{2j}_k} \rho}
\end{align*}

{For all $t_\epsilon>0$ there exist and $l^* > 0$ such that for all $l \geq l^*$ and for each $j$ such that $ t_{j-(N+1)}> t_\epsilon$ and $t_j<=t$ the following hold with probability $1 - \eta$ 

\begin{align*}
    \lvert \mathcal{E}_{\hat{z}_{k+1}, \mathcal{S}^k_j}(x) - L_f^kh(x) \rvert \leq \sqrt{\beta(\rho)} \sigma_k^{2,j} (x) + \alpha(\rho) + \\
    +c_1 N \max 
        \{
            c_2 l^{k} \lVert \varepsilon(t) \rVert_{\infty}, c_3 l^{k-r}  \rVert
        \}
\end{align*}
for all $ x \in \mathcal{T}_{\delta, t_{j-(N+1)}, t_j}$} and $c_i$, $i=1,2,3$ are positive numbers. 
\end{theorem}

\begin{proof}
    \small
    \begin{align*}
        \lVert \mathcal{E}_{\hat{z}_{k+1}, \mathcal{S}^k_j}(x) - L_fh^k(x) \rVert \leq& \\ \leq \lVert \mathcal{E}_{\hat{z}_{k+1}, \mathcal{S}^k_j}(x) - \mathcal{E}_{L_fh^k, \mathcal{S}^k_j}(x) &+ \mathcal{E}_{L_fh^k,\mathcal{S}^k_j}(x) - L_fh^k(x) \rVert \leq \\ 
        \leq \lVert \mathcal{E}_{\hat{z}_{k+1}, \mathcal{S}^k_j}(x) - \mathcal{E}_{L_fh^k, \mathcal{S}^k_j}(x) \rVert  &+  \lVert \mathcal{E}_{L_fh^k, \mathcal{S}^k_j}(x) - L_fh^k(x) \rVert,
    \end{align*}
    then, applying Theorem \ref{THEO:estimate} to the first term of the right side of the above equation, and \lref{lem:MEAN-BOUND} to the second term, the initial statement can be recovered.
\end{proof}



\section{Simulation}%
\label{SEC:SIMULATIONS}  
To test our hypothesis, we consider the same simulation setting proposed in \cite{khan_gaussian_2022}, where the authors used a GP to estimate a safe navigation policy in an unknown environment. We will compare the time derivative of $\hat{h}$ to $\hat{h}^{(1)}$ to see the improvement
\par We have an agent whose state $x \in \R^4$ is described by its position $p \in \R^2$ and its velocity $v \in \R^2$, and whose dynamics is described as a double integrator:
\begin{align*}
    \dot{x} = 
    \left[ 
	\begin{matrix}
	    \dot{p} \\ \dot{v}
	\end{matrix}
    \right]
    = 
    \left[ 
	\begin{matrix}
	    0_{2 \times 2} & I_{2 \times 2} \\
	    0_{2 \times 2} & 0_{2 \times 2}
	\end{matrix}
    \right]
    \left[ 
	\begin{matrix}
	    p \\ v
	\end{matrix}
    \right]
    +
    \left[ 
	\begin{matrix}
	    0_{2 \times 2} \\ I_{2 \times 2}
	\end{matrix}
    \right]
    u
\end{align*}

where $u \in \R^2$ is the control input of the system.
\par Let $\gamma: [0, T] \to \R^4$, with $T > 0$, be an arbitrary desired state trajectory. Let us denote the components of $\gamma$ as
\begin{align*}
    \gamma = 
    \left[ 
	\begin{matrix}
	    p^* \\ v^*
	\end{matrix}
    \right]
\end{align*}
with $p^* \in \R^2$ and $v^* \in \R^2$. 
\par The agent is stabilized on the curve $\gamma$ using the control policy
\begin{align*}
    u = - k_p (p - p^*) - k_v (v - v^*),
\end{align*}
where $k_p, k_v > 0$ are the parameters of the controller.
\par Let $\mathcal{O}_i \subset \R^2$, with $i = 1, \ldots, N$, be $N$ obstacles present in the environment which are described as connected subsets of $\R^2$.
\par For each obstacle let $h_i : \R^2 \to \R$ be a function describing the squared distance of the agent from the obstacle, i.e.
\begin{align*}
    h_i(p) = \min_{y \in \mathcal{O}_i} \lVert p - y \rVert^2
\end{align*}

\par The current distance of the agent from the obstacles is described by the minimum between the $h_i$
\begin{align*}
    h(p) = \min_{i} h_i(p)
\end{align*}

Notice that $h$ is not a smooth function, so to define a differentiable approximation $h_s(p) : \R^2 \to \R$ whose derivative will be estimated by the observer, we rely on the \emph{smooth max function}
\begin{align*}
    h_s(p) = \frac{\sum_{i=1}^{N} h_i(p)e^{\alpha h_i(p)}}
	{\sum_{i=1}^{N} e^{\alpha h_i(p)}}
\end{align*}

which has the property the following properties
\begin{align*}
    \lim_{\alpha \to +\infty} h_s(p) = \max_{i} h_i(p) \\
    \lim_{\alpha \to -\infty} h_s(p) = \min_{i} h_i(p)
\end{align*}

\par In the presented simulation we have set $\alpha = 5$, $k_p = 8$, $k_v = 2$, $l = 20$, $k_1 = 8$, $k_2 = 15$. $\varepsilon(t)$ is a white noise with mean $0$ and variance $0.001$. The results are summarized in the following plots. In \figref{fig:traj} is showed the simulation environment along with the trajectory of the agent. Figure \figref{fig:Value} shows the evolution in time of $L_f \hat{h} = L_f \mathcal{E}_{h_s, S^k_j}$ and $\hat{h}^{(1)} = \mathcal{E}_{\hat{z}_2, S^k_j}$ along with $L_f h_s$. \figref{fig:Error} shows instead the estimation error of both of the regressors. We can clearly see from the plots that, after an initial transient, $\hat{h}^{(1)}$ converges to lower error levels with respect to $L_f \hat{h}$.


\section{Conclusions and future work}\label{SEC:CONCLUSIONS} 
The paper proposed a method  to obtain an analytical estimate of the derivatives of a function along the trajectories of a dynamical system by just measuring the  output and the state of the system. The motivating application scenario is the one in which an autonomous system, is required to navigate in an environment cluttered with obstacles, and high-order barrier functions must be employed to implement control strategies that allow the vehicle to avoid obstacles.
 The output and higher derivatives functions were assumed to be realizations of Gaussian processes and  a high gain observer driven by the measured output was used to construct the data set needed to train the Gaussian processes.  
 We  showed that the proposed method has theoretical guarantees of convergence to the estimate that would have been obtained if the measurements of the derivatives were available.

\par The simulations presented in this article show that in those applications where it is needed an analytic estimate of the derivative of a function of the state of a dynamical system, it is better to fit the Gaussian process with the data generated by an observer instead of taking the derivative of a Gaussian process fitted directly to the sample of the function. 
\par To bring our results to real scenarios there are still many issues to be addressed. For example, we need to reduce the sensitivity to noise of the observer, for instance by using the low-power high-gain observer proposed in \cite{astolfi_low-power_2018}. 
\par Future work will focus on how to leverage this data to improve the safe policy introduced in \cite{khan_gaussian_2022}. Another possible expansion is to study how to incorporate the data generated by the observer directly in a joint estimation of the function and its derivative, as for example done in \cite{raissi_machine_2017}. 

\begin{figure}[ht!]
    \centering
    \includegraphics[scale=0.5]{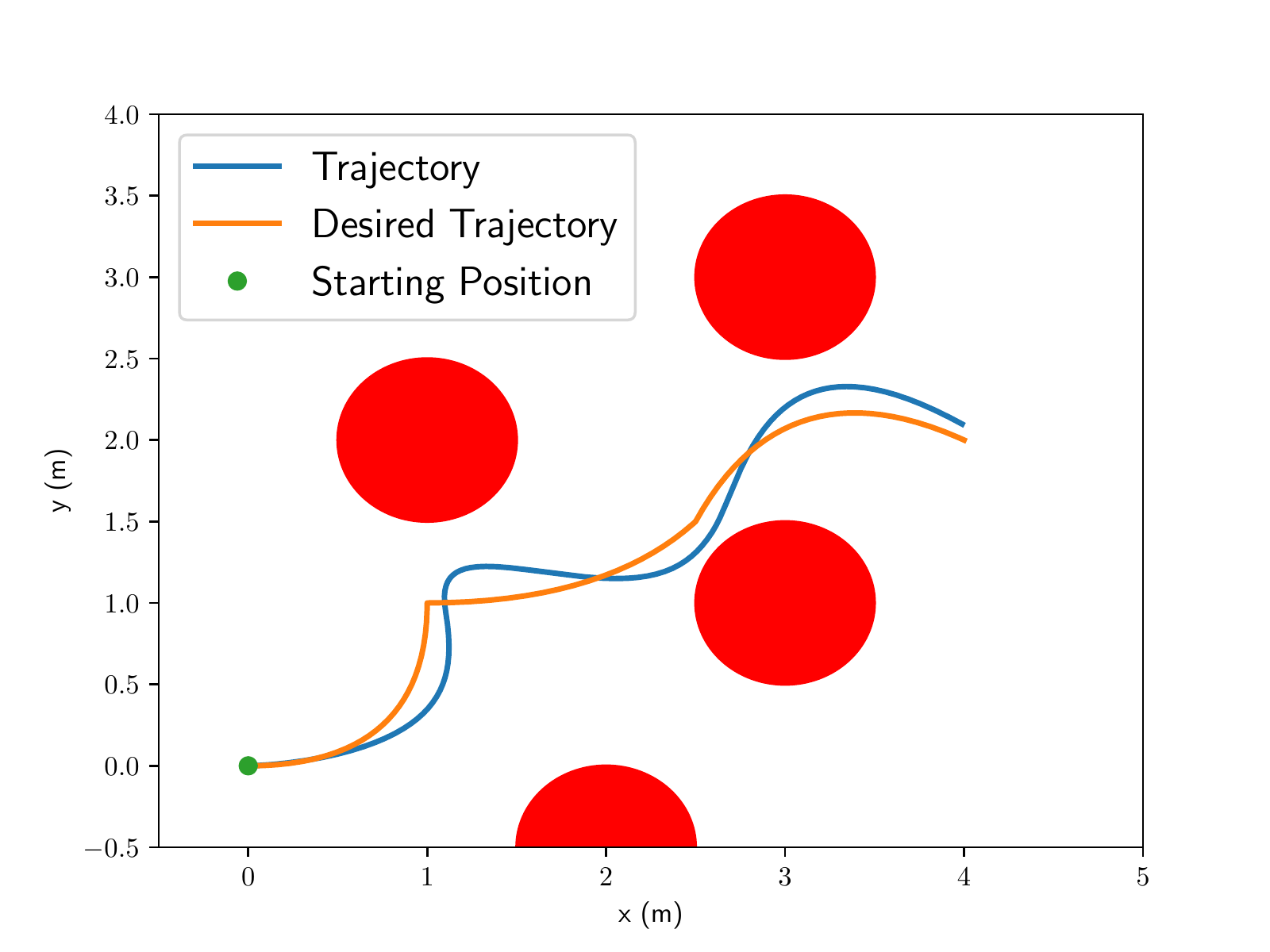}
    \caption{\label{fig:traj} In the figure we can the path followed by the agent in the environment, the desired trajectory was generated using a piece-wise polynomial curve. The red circle are the obstacle in the environment.}
\end{figure}

\begin{figure}[ht!]
    \centering
    \includegraphics[scale=0.5]{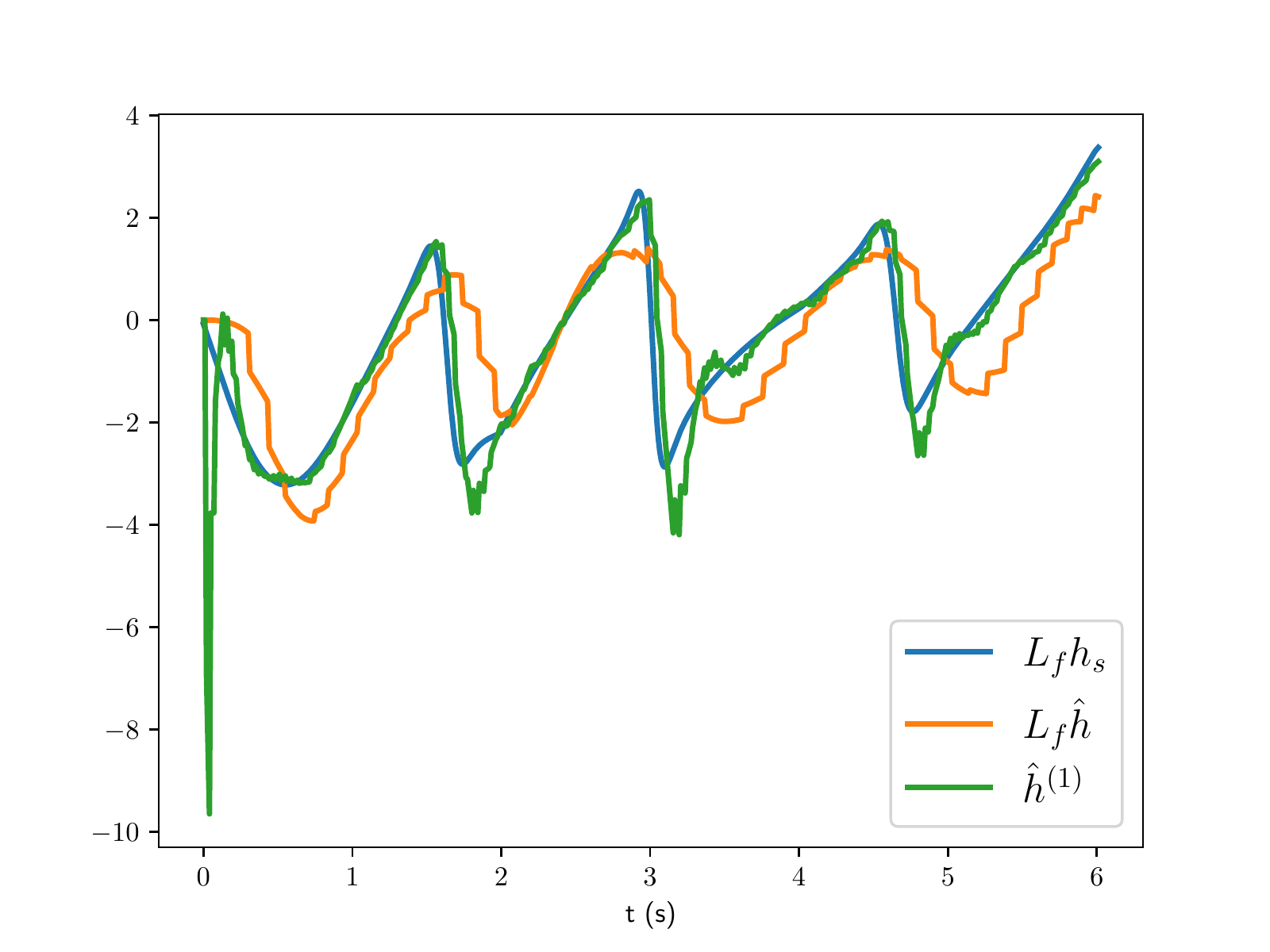}
    \caption{\label{fig:Value} In the figure we can see how $L_f \hat{h}$ and $\hat{h}^{(1)}$ evolve compared to $L_f h_s$. The discontinuities correspond to the sampling instants of the Gaussian processes.}
\end{figure}

\begin{figure}[ht!]
    \centering
    \includegraphics[scale=0.5]{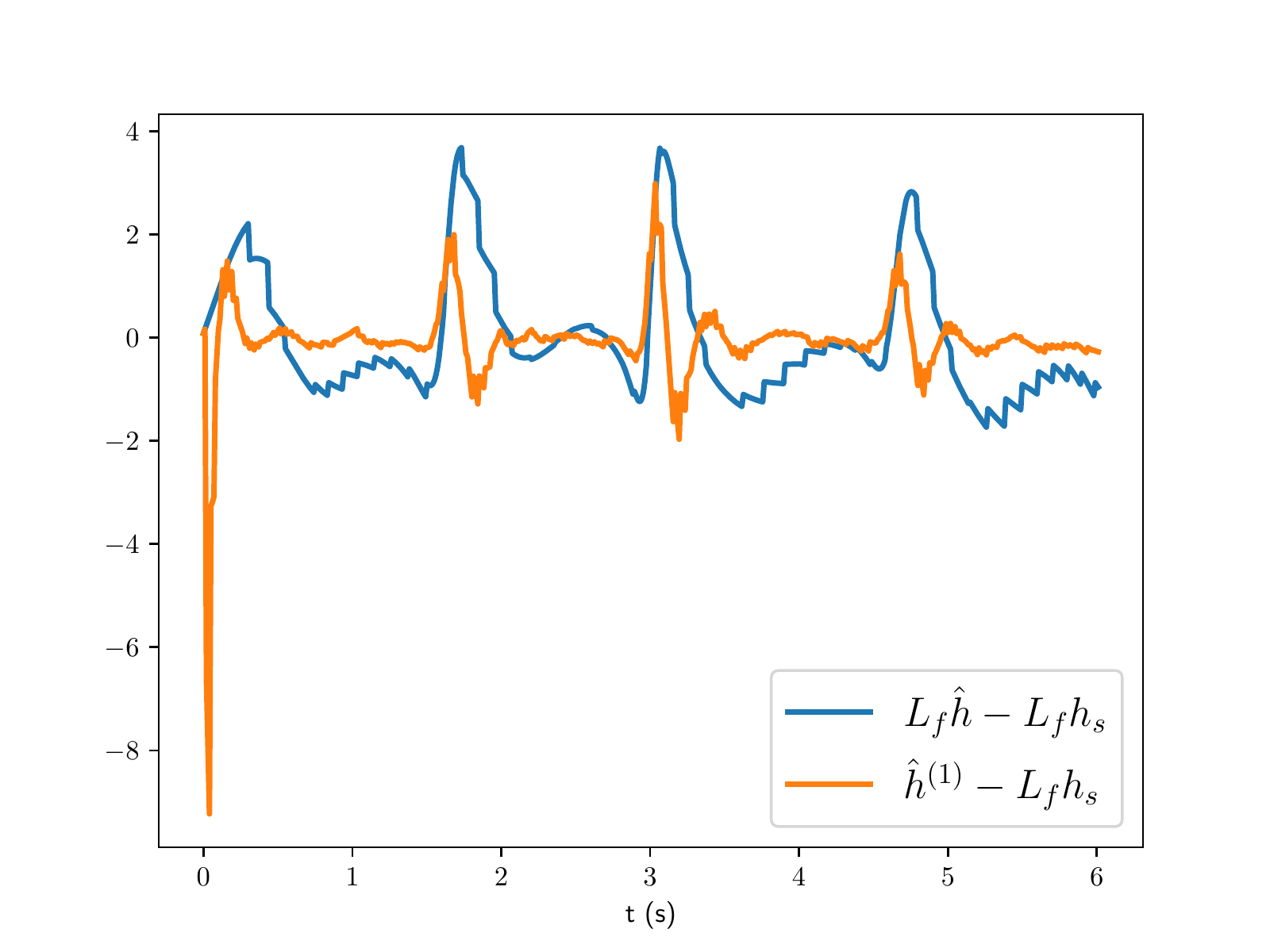}
    \caption{\label{fig:Error}In this figure are reported the estimation error of $L_f \hat{h}$ and $\hat{h}^{(1)}$. We can see that after a small transient, the error of $\hat{h}^{(1)}$ become smaller than $L_f \hat{h}$. The peaks in the plot correspond to a rapid change in the estimated function.}
\end{figure}

\bibliographystyle{unsrt}
\bibliography{ref}

\end{document}